\documentclass[preprint,11pt]{elsarticle}
\pdfoutput=1

\makeatletter
\def\ps@pprintTitle{%
 \let\@oddhead\@empty
 \let\@evenhead\@empty
 \def\@oddfoot{}%
 \let\@evenfoot\@oddfoot}
\makeatother

\usepackage[titletoc,title]{appendix}

\makeatletter
 \@addtoreset{chapter}{part}
 \@addtoreset{@ppsaveapp}{part}
\makeatother

\usepackage{hyperref} 

\usepackage{amsmath,amsthm}
\usepackage{algorithm}
\usepackage{algpseudocode}
\usepackage{setspace}
\floatname{algorithm}{}
\let\Algorithm\algorithm
\renewcommand\algorithm[1][]{\Algorithm[#1]\setstretch{1.2}}

\usepackage{algpseudocode}
\usepackage{euscript}
\usepackage{enumerate}
\usepackage{enumitem}
\setlist{leftmargin=5.5mm}
\usepackage{subfigure}

\usepackage{verbatim} 

\usepackage{xspace} 

\newcommand{\CommentX}[1]{\unskip~#1~}

\algnewcommand{\IIf}[1]{\State\algorithmicif\ #1\ \algorithmicthen}
\algnewcommand{\EndIIf}{\unskip\ \algorithmicend\ \algorithmicif}

\algnewcommand{\IiIf}[1]{\State\algorithmicif\ #1\ \algorithmicthen}
\algnewcommand{\EndIiIf}{\unskip\ \algorithmicend\ \algorithmicif}

\algnewcommand{\IfThenElse}[3]{
 \State \algorithmicif\ #1\ \algorithmicthen\ #2\ \algorithmicelse\ #3}
\algnewcommand{\EndIfThenElse}{\unskip\ \algorithmicend\ \algorithmicif}

\newtheorem{theorem}{Theorem}
\newtheorem{prop}[theorem]{Proposition}

\begin{document}

\begin{frontmatter}

\title{A new exact approach for the Bilevel Knapsack with Interdiction Constraints}

\date{}

 \author[1,2]{Federico Della Croce}
 \author[1]{Rosario Scatamacchia}
 
  \address[1]{\small Dipartimento di Ingegneria Gestionale e della Produzione, Politecnico di Torino,\\ Corso Duca degli Abruzzi 24, 10129 Torino, Italy, \\{\tt \{federico.dellacroce, rosario.scatamacchia\}@polito.it }}
 \address[2]{CNR, IEIIT, Torino, Italy}
 
\begin{abstract}We consider the Bilevel Knapsack with Interdiction Constraints, an extension of the classic 0-1 knapsack problem formulated as a Stackelberg game 
with two agents, a leader and a follower, that choose items from a common set and hold their own private knapsacks. First, the leader selects some items to be interdicted for the follower while satisfying a capacity constraint. Then the follower packs a set of the remaining items according to his knapsack constraint in order to maximize the profits. The goal of the leader is to minimize the follower's profits. The presence of two decision levels makes this problem very difficult to solve in practice: the current state-of-the-art algorithms can solve to optimality instances with 50-55 items at most. 
We derive effective lower bounds and present a new exact approach that exploits the structure of the induced follower's problem. The approach successfully solves all benchmark instances within one second in the worst case and larger instances with up to 500 items within 60 seconds.
\end{abstract}

\end{frontmatter}
\section{Introduction}
\label{sec:Intro}
In the recent years, a growing attention has been centered to multilevel programming. This emerging field considers optimization problems with a hierarchal structure where many decision makers sequentially operate to reach conflicting objectives. Each agent takes decisions that may affect objectives and decisions of the agents at lower levels. At the same time, the latter decisions impact on the objectives of the agents at upper levels. Hierarchal contexts arise in many real-life applications in supply chains, energy sector, logistics and telecommunication networks among others. The presence of many decision levels makes these problems very challenging to solve. \\
The most relevant research in the field has been pursued for bilevel optimization where two agents, denoted as a leader and a follower, play a Stackelberg game (\cite{St52}).
In this game, the leader takes the first decision and then the follower reacts taking into account the leader's strategy. Eventually, the agents receive a pay-off which depends on both leader's and follower's choices. The goal is typically to find a strategy for the leader that optimizes his own objective.
Two standard assumptions are considered in a Stackelberg game: \textit{perfect knowledge}, that is each agent knows the problem solved by the other agent; \textit{rationale behavior}, namely each agent has no interest in deviating from his own objective. \\
Bilevel optimization considers Mixed-Integer Bilevel Linear Programs (MIBLP) where both the leader and the follower solve a combinatorial optimization problem with linear objective function and constraints and with either continuous or integer variables. The first generic Branch and Bound approach for MIBLP was provided in \cite{BaMo90}. Branch and Cut schemes were introduced in \cite{DeRa09}, \cite{Den11}. Further approaches were proposed in \cite{CaMa15, FiLjMoSi16,XuWa14}. An improved generic MIBLP solver has been recently proposed in \cite{FiLjMoSi17}. We refer to \cite{FiLjMoSi17} and the references therein for an overview on MIBLP solvers and related applications.\\

In this paper, we consider the Bilevel Knapsack with Interdiction Constraints (BKP), as introduced in \cite{Den11}.
The problem is an extension of the classic 0-1 Knapsack Problem (KP) (see monographs \cite{KePfPi04} and \cite{MarTot90}) formulated as a Stackelberg game. 
More precisely, the leader and the follower choose items from a common set and hold their own private knapsacks. First, the leader selects some items to be interdicted for the follower while satisfying a capacity constraint. Then the follower packs a set of the remaining items according to his knapsack constraint in order to maximize the profits. The goal of the leader is to minimize the follower's profits. \\
In \cite{CaCarLoWo13} it is shown that BKP is $\Sigma_{2}^p$-complete in the polynomial hierarchy complexity. Essentially, BKP cannot be formulated as a single level problem unless the polynomial hierarchy collapses (as also pointed out in \cite{CaCarLoWo16}). This makes the problem even more difficult to solve than an NP-Complete problem. We refer to \cite{Je85} for an introduction on polynomial hierarchy. \\
One of the best performing algorithms for BKP is given in \cite{CaCarLoWo16}. The algorithm, denoted as CCLW, relies on the dualization of the continuous relaxation of the follower's problem and on iteratively computing upper bounds for the problem until a stopping criterion applies. The approach is motivated by the lack of significant lower bounds for the problem. 
Algorithm CCLW solves to optimality instances with 50 items within a CPU time limit of 3600 seconds, running out of time in instances with 55 items only. Very recently, an improved branch-and-cut algorithm has been given in \cite{FiLjMoSi18}. The proposed approach manages to solve to optimality all benchmark instances in \cite{CaCarLoWo16}, requiring at most a computation time of about 85 seconds in an instance with 55 items. However, no computational evidence is provided in \cite{FiLjMoSi18} about the performance of the derived algorithm on larger instances. We also mention the work of \cite{FiMoSiEJOR18} where a heuristic approach is proposed for BKP and for other interdiction games. 


\smallskip
Other bilevel knapsack problems have been tackled in the literature. We mention the work in \cite{BrHaMa13} where the leader cannot interdict items but modifies the follower's capacity. In \cite{ChZh13}, the leader can modify the follower's objective function only. As discussed in \cite{CaCarLoWo16}, these knapsack problems are easier to handle than BKP. Recently, a polynomial algorithm has been provided in \cite{CaLoMa18} for the BKP variation where the follower solves a continuous knapsack problem. \\

\medskip
Our contribution for BKP is twofold. 
First, we derive effective lower bounds based on mathematical programming. 
Second, we present a new exact approach that exploits the induced follower's problem and the derived lower bounds. The proposed approach shows up to be very effective 
successfully solving all benchmark literature instances provided in \cite{CaCarLoWo16} within few seconds of computation. Moreover, our algorithm manages to solve to optimality instances with up to 500 items within a CPU time limit of 60 seconds. 

\smallskip

The paper is organized as follows. In Section~\ref{sec:themodel}, the bilevel linear programming formulation of the problem is introduced. In Section \ref{sec:LB}, we discuss the lower bounds for BKP. We outline the proposed exact solution approach in Section \ref{sec:ExaApp} and discuss the computational results in Section \ref{sec:ComRes}. 
Section~\ref{sec:Concl} provides some concluding remarks.



\section{Notation and problem formulation}
\label{sec:themodel}

In BKP a set of $n$ items and two knapsacks are given. Each item $i$ $(=1,\dots, n)$ has associated a profit $p_i > 0$ and a weight $w_i > 0$ for the follower's knapsack and a weight $v_i > 0$ for the leader's knapsack. 
Leader and follower have different knapsack capacities denoted by $C_u$ and $C_l$, respectively. Quantities $p_i$, $v_i$, $w_i$ $(i=1,\dots, n)$, $C_u$, $C_l$ are assumed to be integer, with $v_i < C_u$ and $w_i < C_l$ for all $i$. To avoid trivial instances, it is also assumed that $\sum\limits_{i=1}^{n} v_{i} > C_u$ and $\sum\limits_{i=1}^{n} w_{i} > C_l$.
We introduce $0/1$ variables $x_i$  $(i=1, \dots,n)$ 
equal to one if the leader selects items $i$ and $0/1$ variables $y_i$ equal to one if item $i$ is chosen 
by the follower. BKP can be modeled as follows: 



\begin{align}
\text{min}\quad & \sum\limits_{i=1}^{n} p_{i}y_{i} \label{eq:ObjL}\\
	\text{subject to}\quad
	& \sum\limits_{i=1}^{n} v_{i}x_{i} \leq C_u \label{eq:capL}\\
	& x_{i} \in\{0,1\} \qquad i= 1,\dots,n \label{eq:varDefX} \\
	\text{where $y_1, \dots, y_n$ solve} \nonumber\\
	\text{the follower's problem:}\quad
	\text{max}\quad & \sum\limits_{i=1}^{n} p_{i}y_{i} \label{eq:ObjF}\\
	\text{subject to}\quad
	& \sum\limits_{i=1}^{n} w_{i}y_{i} \leq C_l \label{eq:capF}\\
	&y_{i} \leq 1 - x_{i} \qquad i= 1,\dots,n \label{eq:Interd} \\
	&y_{i} \in\{0,1\} \qquad i= 1,\dots,n \label{eq:varDefY} 
\end{align}



The leader's objective function \eqref{eq:ObjL} minimizes the profits of the follower through the interdiction constraints  \eqref{eq:Interd}. These constraints ensure that each item $i$ can be selected by the follower, i.e. $y_i \leq 1$, only if the item is not interdicted by the leader, i.e. $x_i = 0$. Constraint \eqref{eq:capL} represents the leader's capacity constraint. The objective function \eqref{eq:ObjF} maximizes the follower's profits and constraint \eqref{eq:capF} represents the follower's capacity constraint. Constraints \eqref{eq:varDefX} and \eqref{eq:varDefY} define the domain of the variables.

The optimal solution value of model \eqref{eq:ObjL}-\eqref{eq:varDefY} is denoted by $z^*$. The optimal solution vectors of variables $x_i$ and $y_i$ are respectively denoted by $x^*$ and $y^*$. Notice that in model \eqref{eq:ObjL}-\eqref{eq:varDefY} there always exists an optimal solution for the leader which is maximal, namely where items are included in the leader's knapsack until there is no enough capacity left. 

\smallskip 
Let us now recall the optimal solution of the continuous relaxation of a standard KP, namely the follower's model \eqref{eq:ObjF}-\eqref{eq:varDefY} without constraints \eqref{eq:Interd} and constraints \eqref{eq:varDefY} replaced by inclusion in $[0, 1]$. Under the assumption $\sum\limits_{i=1}^{n} w_{i} > C_l$, this solution has the following structure. Consider the sorting of the items by non-increasing ratios of profits over weights:


\begin{equation}
\label{eq:SortByEff}
\frac{p_1}{w_1} \geq \frac{p_2}{w_2} \geq \dots \geq \frac{p_n}{w_n}.
\end{equation}


According to this order, items $j = 1,2,\dots$ are inserted into the knapsack as long as $\sum\limits_{k=1}^{j} w_{k} \leq C_l$. The first item $s$ which cannot be fully packed is commonly denoted in the knapsack literature as the \textit{split} item (or \textit{break}/\textit{critical} item). 
The optimal solution of the KP linear relaxation is given by setting $y_j = 1$ for  $j= 1,\dots,s-1$, $y_j = 0$ for  $j= s+1,\dots,n$ and $y_s = (C_l - \sum\limits_{j=1}^{s-1} w_j)/w_s$. The solution with items $1, \dots, (s-1)$ is a feasible solution for KP and is commonly denoted as the \textit{split solution}.

In the remainder of the paper, we assume the ordering of the items \eqref{eq:SortByEff}. 
We denote by $KP(x)$ the follower's knapsack problem induced by a leader's strategy encoded in vector $x$, i.e. a knapsack problem with item set $$S:=\{i: x_i = 0, x_i \in x \}.$$ 
We also denote by $KP^{LP}(x)$ the corresponding Linear Programming (LP) relaxation. If $\sum_{i \in S} w_i > C_l$,  we define the \textit{critical} item $c$ of $KP^{LP}(x)$ as the last item with a strictly positive value in its optimal solution. Thus, we have $y_c \in (0, 1]$ and a corresponding split solution with profit
\begin{equation}
\label{eq:SplitSol}
\sum\limits_{i \in  S: i < c} p_i = \sum\limits_{i = 1}^{c-1} p_i(1 - x_i)
\end{equation}
which constitutes a feasible solution for $KP(x)$.
Notice that 
we denote by $z(M)$ the optimal solution value of any given
mathematical model $M$. 

\section{Computing lower bounds on BKP}
\label{sec:LB}

Consider the optimal solution vector $x^*$. In the induced follower's knapsack problem $KP(x^*)$ with item set $S$, two cases can occur:
either there is no critical item in $KP^{LP}(x^*)$, namely $\sum_{i \in S} w_i  \leq C_l$, or one critical item exists, namely $\sum_{i \in S} w_i > C_l$. The first case can be easily handled by considering that the follower will pack all items not interdicted by the leader. This case is discussed in Section~\ref{NOcrit}. \\
In the second case, we derive effective lower bounds on BKP that constitute the main ingredient of the exact approach presented in Section~\ref{sec:ExaApp}. Since we don't know a priori the leader's optimal solution $x^*$, we proceed by guessing the critical item of  $KP^{LP}(x^*)$, namely
we formulate an Integer Linear Programming (ILP) model where we impose that a given item $c$ must be critical and evaluate the profit of the corresponding split solution. We consider binary variables $k_j$ $(j= 1, \dots, w_c)$ associated with the weight contribution of the critical item and introduce the following model (denoted as $CRIT_1(c)$).

\smallskip
{$\quad CRIT_1(c)$:}


\begin{align}
\text{min}\quad & \sum\limits_{i=1}^{c-1} p_{i}(1-x_i) \label{eq:ObjCRIT}\\
	\text{subject to}\quad
	& \sum\limits_{i=1}^{n} v_{i}x_{i} \leq C_u \label{eq:capLeadCRIT}\\
	& \sum\limits_{i=1}^{c-1} w_{i}(1-x_{i}) + \sum\limits_{j=1}^{w_c} jk_j = C_l \label{eq:capFollCRIT}\\
	& \sum\limits_{j=1}^{w_c} k_j = 1 \label{eq:capSOSCRIT}\\
	& x_c = 0 \label{eq:NotInterCRIT}\\
	& x_{i} \in\{0,1\} \qquad i= 1,\dots,n \label{eq:varDefXCRIT} \\
	& k_{j} \in\{0,1\} \qquad j= 1,\dots,w_c \label{eq:varDefKCRIT} 
\end{align} 


The objective function \eqref{eq:ObjCRIT} minimizes the value of the split solution. Constraint \eqref{eq:capLeadCRIT} represents the leader's capacity constraint. Constraints \eqref{eq:capFollCRIT} and \eqref{eq:capSOSCRIT} ensure that item $c$ is critical as it is the last item packed, with a weight in the interval $[1, w_c]$. Constraint \eqref{eq:NotInterCRIT} indicates that item $c$ can be critical only if it is not interdicted by the leader. Constraints \eqref{eq:varDefXCRIT} and \eqref{eq:varDefKCRIT} indicate that all variables are binary. We can state the following proposition.

\begin{prop}
\label{firstProp}
If there exists a critical item $c$ in $KP^{LP}(x^*)$, then $z(CRIT_1(c))$ is a valid lower bound on $z^*$.
\end{prop}
\begin{proof}
Under the assumption that item $c$ is critical in $KP^{LP}(x^*)$, the optimal BKP solution $x^*$ constitutes a feasible solution for model $CRIT_1(c)$. Let denote by $z_1$ the corresponding solution value that coincides with the value of the split solution in $KP(x^*)$. Since the follower maximizes the profits in $KP(x^*)$ obtaining a solution with a value greater than (or equal to) the one of the split solution, we have $z_1 \leq z^*$. But this means that there exists an optimal solution of model $CRIT_1(c)$ such that $z(CRIT_1(c)) \leq z_1$ which implies a lower bound on $z^*$. 
\end{proof}

The previous proposition already 
provides a first significant lower bound for the problem. However, following the reasoning in the proof of Proposition \ref{firstProp}, we 
remark that improved bounds on $z^*$ can be derived by considering any feasible
solution for $KP(x^*)$ that might be obtained by removing (adding) items that were not interdicted by the leader and that were
selected (not selected) by the split solution, provided that the follower capacity is not exceeded.
Indeed, this corresponds to removing tuples of items $i \in [1,c-1]: x_i=0$ and/or to adding tuples of items $i \in [c,n]: x_i=0$
from the split solution without exceeding the follower capacity.

Notice that, the state-of-the-art algorithms for KP, \textit{Minknap} (\cite{Pis97}) and \textit{Combo} (\cite{MarPisTot99}) consider that in general only few items with ratio $p_i/w_i$ close to that of the critical item change their values in an optimal solution with respect to the values taken in the split solution. These items constitute the so-called \textit{core} of the knapsack. 
\textit{Minknap} and \textit{Combo} start with the computation of the split solution and an expanding core initialized with the critical item only. Then, the algorithms iteratively enlarge the core by evaluating both the removal of items from the split solution and the addition of items after the critical item. The empirical evidence illustrates that an optimal (or close to be optimal) KP solution is typically found after few iterations. 

We cannot precisely characterize the features of these exact algorithms by a set of constraints within an ILP model, but we can mimic the same algorithmic reasoning by considering subsets of the items set $c-\delta,...,c+\delta$ including the critical item $c$ for any given core size $2\delta+1$. In each subset, the items $i: i \leq c-1$ are removed from the split solution, while
the items $j: j \geq c$ are added to the solution. Correspondingly, the initial profit and weight of the split solution are modified by subtracting the profits and the weights of the removed items and by summing up the profits and the weights of the added items. 

Then, for any given subset $\tau$ of the items set $c-\delta,...,c+\delta$, let 
 $p^{\tau}$ and $w^{\tau}$ be the overall profit (namely the value of the improvement upon the split solution) and weight contributions of the items in $\tau$, namely:
\begin{align}
p^{\tau} = -\sum\limits_{i \in \tau: i < c} p_i + \sum\limits_{j \in \tau: j \geq c} p_j; \label{pT}\\
w^{\tau} = -\sum\limits_{i \in \tau: i < c} w_i + \sum\limits_{j \in \tau: j \geq c} w_j. \label{wT}
\end{align}

A subset $\tau$ with $p^\tau \leq 0$ is not considered since it does not improve upon the split solution.
Instead, an improving subset with $p^\tau > 0$ is feasible only if $w^{\tau} \leq w_c$ and all items in $\tau$ are not interdicted by the leader.  In that case, by keeping the notation of model $CRIT_1(c)$, an improvement $\pi$ can be determined if the following 
constraint is added:

\begin{equation}
\label{TupleCons}
\pi \geq p^\tau(\sum\limits_{j = \max\{1;w^\tau\}}^{w_c} k_j - \sum\limits_{i \in \tau} x_i). 
\end{equation}

Correspondingly, a new model can be generated by 
introducing a non-negative variable $\pi$ that carries the maximum additional profit to the split solution value provided by any
of the additional constraints (\ref{TupleCons}) indicated above. 
These 
constraints, denoted as $\mathcal{F}(\pi, x, k)$, link variable $\pi$ to variables $x_i$ and $k_j$. The model (denoted as $CRIT_2(c)$) is as follows.\\

{$\quad CRIT_2(c)$:}
\begin{align}
\text{min}\quad & \sum\limits_{i=1}^{c-1} p_{i}(1-x_i) + \pi \label{eq:ObjCRIT2}\\
	\text{subject to}\quad
	& \mathcal{F}(\pi,x,k) \label{genConstonPi2} \\
	& \eqref{eq:capLeadCRIT}, \eqref{eq:varDefKCRIT} \tag*{} \\ 
	& \pi \geq 0 \label{eq:varDefPi}  
\end{align}  


Clearly, due to the addition of constraints in $\mathcal{F}(\pi,x,k)$, for any $c$ we have $z(CRIT_1(c)) \leq  z(CRIT_2(c))$. 
Notice that, in all these additional constraints,
only items which will not be interdicted by the leader can be packed and the follower's capacity constraint is not violated. We denote as \textit{proper} any set $\mathcal{F}(\pi,x,k)$ that satisfies both conditions.
After the set $\mathcal{F}(\pi, x, k)$ is built, variable $\pi$ will carry the maximum profit obtainable
in addition to the profit of the split solution.
\begin{prop}
\label{secTh}
If $KP^{LP}(x^*)$ admits a critical item $c$ and model $CRIT_2(c)$ has a proper set $\mathcal{F}(\pi, x, k)$, then $z(CRIT_2(c)) \leq z^*$.
\end{prop}
\begin{proof}
Since model $CRIT_2(c)$ considers feasible solutions for $KP(x^*)$, the inequality holds by applying the same argument 
of Proposition \ref{firstProp}. 
\end{proof}

\section{A new exact approach for BKP}
\label{sec:ExaApp}

\subsection{Overview}
We propose an exact algorithm for BKP that considers the possible existence of a critical item in $KP^{LP}(x^*)$ and exploits the bounds provided by model $CRIT_2(c)$. The approach involves two main steps. In the first step, the possible non-existence of a critical item is first evaluated. Then, the approach assumes the existence of a critical item and identifies a set of possible candidate items. For each candidate item $c$ and a parameter $\delta$ to identify the core size, model $CRIT_2(c)$ is built by considering several subsets of additional constraints (\ref{TupleCons}).
Then the linear relaxation $CRIT_2^{LP}(c)$ is solved, where the integrality constraints \eqref{eq:varDefXCRIT} and \eqref{eq:varDefKCRIT} are replaced by inclusion in $[0,1]$.

The feasible problems $CRIT_2^{LP}(c)$ are sorted by increasing optimal value so as to identify an order of the most promising subproblems to explore. A limited number of feasible BKP solutions is also computed in this step. 

In the second step, each relevant subproblem is explored by constraint generation until the subproblem can be pruned. An optimal BKP solution is eventually returned. 
The approach takes as input five parameters $\alpha$, $\beta$, $\delta$, $\mu$, $\gamma$ and relies on an ILP solver along its steps. We discuss the steps of the algorithm in the following. The corresponding pseudo code is then provided. 
 
\subsection{Step 1}
\label{TheStep1}
\subsubsection{Handling the possible non-existence of a critical item}
\label{NOcrit}
We first consider the case where there does not exist a critical item in $KP^{LP}(x^*)$. Thus, the follower will select all available items which are not interdicted by the leader and an optimal solution of BKP is found by solving the following problem $NCR$.\\

\smallskip
{$\quad NCR$:}
\begin{align}
\text{min}\quad & \sum\limits_{i=1}^{n} p_{i}(1-x_{i})  \label{eq:ObjFirst}\\
	\text{subject to}\quad
	& \sum\limits_{i=1}^{n} v_{i}x_{i} \leq C_u \label{eq:capLFirst}\\
	& \sum\limits_{i=1}^{n} w_{i}(1-x_{i}) \leq C_l \label{eq:capFWmax}\\
& x_{i} \in\{0,1\} \qquad i= 1,\dots,n \label{eq:varDefXFirst}
\end{align}
 
If problem $NCR$ is feasible, let denote by $x'$ the related optimal solution 
representing the leader's strategy. The corresponding follower's solution is denoted by $y'$, with $y'_i = 1 - x'_i$ $(i=1, \dots, n)$. 
The current best solution $(x^*, y^*)$ with value $z^*$ (which will be optimal at the end of the algorithm) is initialized accordingly (Lines \ref{NoC}-\ref{endNoC} of the pseudo code). 

\subsubsection{Identifying the relevant critical items}
\label{criticalitems}
We now assume that there exists a critical item $c$ in $KP^{LP}(x^*)$  (Lines \ref{CritItems}-\ref{endCritItems}) and estimate the first and last possible items $l$ and $r$ that can be critical according to ordering \eqref{eq:SortByEff}. For item $l$ we have 
\begin{align}
\label{comL}
l := \min \{ j: \sum\limits_{i=1}^{j} w_{i} \geq C_l\}. 
\end{align}
All items $1, \dots, (l-1)$ cannot in fact be critical even without the leader's interdiction. 
For the last item $r$, we first compute the maximum weight of the follower that can be interdicted by the leader (similarly as in \cite{CaCarLoWo16}) by solving the following problem (denoted by $LW$).\\

\smallskip
{$\quad LW:$}
\begin{align}
\text{max}\quad & \sum\limits_{i=1}^{n} w_{i}x_{i}  \label{eq:ObjWmax}\\
	\text{subject to}\quad
	& \sum\limits_{i=1}^{n} v_{i}x_{i} \leq C_u \label{eq:capLWmax}\\
& x_{i} \in\{0,1\} \qquad i= 1,\dots,n \label{eq:varDefXWmax}
\end{align}
Item $r$ is defined as
\begin{align}
\label{comR}
r := \min \{ j: \sum\limits_{i=1}^{j} w_{i} \geq C_l + z(LW)\}. 
\end{align}
Since from \eqref{comR} we have$\sum\limits_{i=1}^{r} w_{i}(1 -x_i) \geq C_l$ for any leader's strategy, all items from $(r+1)$ to $n$ cannot be critical.

\subsubsection{Building models $CRIT_2(c)$} 
\label{BuildMod}
For each candidate critical item $c \in  [l,r]$, we formulate model $CRIT_2(c)$ by constructing a proper set $\mathcal{F}(\pi, x, k)$ as follows. 
Consider the subsets involving items in the interval $[c - \delta, c + \delta]$. 
Even for small value of $\delta$, the number of subsets can be very large. Hence, in order to 
limit the number of constraints in $\mathcal{F}(\pi, x, k)$,
we propose a different strategy that greedily selects the subsets according to 
the procedure denoted as $ComputeTuples$ and sketched below. 

For a given value of $\delta$, we consider the interval of items $[a, b]$, with 
$a=\max\{1;c-\delta\}$ and $b=\min\{c+\delta;n\}$. Starting by the empty set, we enumerate at most $\alpha$ ``backward'' sets with items $(c-1), \dots, a$ in increasing order of size. Each set has a profit and weight equal to the sum of profits and weights of the included items. We also compute at most $\beta$ ``forward'' sets with items $c, \dots, b$ in increasing order of size and with a weight not superior to the maximum weight of a backward set. This in order to exclude forward sets having less chance to be combined with a backward set. 

Then the backward (resp. forward) sets are ordered by increasing (resp. decreasing) profit. We combine each backward set with a forward set and generate a tuple $\tau$. If $p^{\tau} > 0$ and $w^{\tau} \leq w_c$, we add constraint \eqref{TupleCons} to $\mathcal{F}(\pi, x, k)$. We continue adding constraints to $\mathcal{F}(\pi, x, k)$ until their number is 
superior to an input parameter $\mu$. If not previously included, we also add to set $\mathcal{F}(\pi, x, k)$ the constraint $\pi \geq p_c k_{w_c}$ which handles the possible adding of the critical item to the split solution if the residual capacity is equal to $w_c$. 

	\begin{algorithm}
		\label{algo:Tuples}
	\caption{\textit{ComputeTuples($c$, $\alpha$, $\beta$, $\delta$, $\mu$)}}
	\begin{algorithmic}[1]
	\State{Consider items in the interval $[a,b]$ with $a := \max \{c -\delta; 1 \}$, $b := \min \{c +\delta; n\}$}.
	\State{Starting from the empty set and in increasing order of size, enumerate $\alpha$ backward sets with items $(c-1), \dots, a$. Denote by $w_{max}$ the maximum weight of a backward set. Order the sets by increasing profits.}
	\State{Enumerate $\beta$ forward sets with items $c, \dots, b$ in increasing order of size and with a weight not superior to $w_{max}$. Order the sets by decreasing profits.}
	\State{Take the first available backward set. Merge the set with a forward set and generate tuple $\tau$. \label{mergeSets}}
	\State{If  $p^\tau > 0$ and $w^\tau \leq w_c$, add constraint $\pi \geq p^\tau(\sum\limits_{j = \max\{1;w^\tau\}}^{w_c} k_j - \sum\limits_{i \in \tau} x_i)$ to $\mathcal{F}(\pi, x, k)$. \label{addF}}
	\State{Iterate Steps \ref{mergeSets}-\ref{addF}  as long as $|\mathcal{F}(\pi, x, k)| \leq \mu$.}
 \State{If not already included, add to $\mathcal{F}(\pi, x, k)$ constraint $\pi \geq p_c k_{w_c}$.}
\end{algorithmic}
\end{algorithm}

\smallskip
Then we solve models $CRIT_2^{LP}(c)$ for each $c \in [l,r]$ and order the models by increasing optimal value so as to have an order of most promising subproblems to explore. If for the first subproblem we have $z(CRIT_2^{LP}(c)) \geq z^*$, an optimal BKP solution is already certified (Line \ref{endCritItems} of the pseudo code).

\subsubsection{Computing feasible BKP solutions} 
\label{heursol}
 According to the previous order of subproblems, we compute BKP feasible solutions by considering the first $\gamma$ subproblems (Lines \ref{FeasGammaSol}-\ref{endFeasGammaSol}). For a given item $c$, we solve model $CRIT_2(c)$ obtaining a solution $\hat{x}$. \\
If $z(CRIT_2(c)) < z^*$, we solve the induced follower's problem $KP(\hat{x})$ with optimal solution $\hat{y}$ and update the current best solution if $z(KP(\hat{x})) < z^*$.

\subsection{Step 2}
\label{TheStep2}
This step consider all relevant (ordered) suproblems $CRIT_2(c)$.
For each subproblem, we first test for standard variables fixing and then
%
each subproblem is explored by means of a constraint generation approach (Lines \ref{solveSub}-\ref{endsolveSub}).
 
\subsubsection{Fixing variables in subproblems}
\label{fixvar}
For a given problem $CRIT_2^{LP}(c)$, denote the optimal values of variables $x_{i}$ and $k_j$ by $x_{i}^{LP}$ and $k_j^{LP}$ respectively. 
Let $r_{x_i}$ and $r_{k_j}$ be the reduced costs of non basic variables in the optimal solution of $CRIT_2^{LP}(c)$. We apply then standard variable-fixing techniques from Integer Linear Programming:
if the gap between the best feasible solution available and the optimal solution value of the continuous relaxation solution is not greater than the absolute value of a non basic variable reduced cost, then the related variable can be fixed to its value in the continuous relaxation solution. Thus, 
the following constraints are added to $CRIT_2(c)$:

\begin{align}
& \forall\, i: |r_{x_i}| \geq z^* - z(CRIT_2^{LP}(c)) , \quad x_i = x_i^{LP}; \label{eq:rxi}\\
& \forall\, j: |r_{k_j}| \geq  z^* - z(CRIT_2^{LP}(c)), \quad k_j = k_j^{LP}.\label{eq:rkj}
\end{align} 

\subsubsection{Solving subproblems}
\label{solvesub}
For each open subproblem $CRIT_2(c)$, we first solve $CRIT_2(c)$  obtaining a solution $\bar{x}$. If the corresponding objective value is lower than the current best feasible solution value, we solve $KP(\bar{x})$ with solution $\bar{y}$ and if an improving solution is found, the current best solution is updated, as in Section~\ref{heursol}. 
Then, we 
add to $CRIT_2(c)$ constraints

\begin{align}
\sum\limits_{i : \bar{x_i} = 0}^{n} x_i + \sum\limits_{i : \bar{x_i} = 1}^{n} (1-x_i) \geq 1; \label{changeone} \\
\sum\limits_{i : \bar{y_i} = 1}^{n} x_i \geq 1. \label{interdone}
\end{align}

These cuts impose that at least one variable $x_i$ in solution vector $\bar{x}$ must be discarded (constraint \eqref{changeone}) and at least one item selected by the follower in solution $\bar{y}$ must be interdicted (constraint \eqref{interdone}). We solve $CRIT_2(c)$ with two more constraints and apply the same procedure until $z(CRIT_2(c)) \geq z^*$ or the problem becomes infeasible. 
At the end of Step 2, the optimal BKP solution $(x^*, y^*)$ is returned (Line \ref{returnOpt}). 

\begin{algorithm}[H]
\setstretch{1.00}
		\label{algo:ExAp}
	\caption{Exact solution approach}
	\begin{algorithmic}[1]
	\State{\textbf{Input:} BKP instance, parameters $\alpha$, $\beta$, $\delta$, $\mu$, $\gamma$.}
	
	\Comment{Step 1}			
\State{\textit{Handle the absence of a critical item:}}			
%
				\State{solve $NCR$; $z^* \leftarrow$ $+\infty$; \label{NoC}}		 		
				\IIf{\textit{$NCR$ has a feasible solution} }
				$x^* = x'$, $y^* = y'$, $z^*= z(NCR)$;
				\EndIIf 		\label{endNoC}	
		\CommentX{} 
		
\State{\textit{Identify the candidate critical items and build models $CRIT_2(c)$:}   \label{CritItems}}	
				    \State{Compute the interval of critical items $[l,r]$: $l \leftarrow$ apply \eqref{comL}, $r \leftarrow$ apply \eqref{comR};}
					\For {all $c$ in $[l,r]$} 
					
					\State{Build model $CRIT_2(c)$  by procedure \textit{ComputeTuples($c$, $\alpha$, $\beta$, $\delta$, $\mu$)};} 
					\State{Solve model $CRIT_2^{LP}(c)$;}
										
					\EndFor 	

\State{Sort models $CRIT_2(c)$ by increasing $z(CRIT_2^{LP}(c))$.} 
\State{$\Longrightarrow$ Create a list of ordered critical items $L = \{c_1, c_2, \dots \}$;}
\IIf{$z(CRIT_2^{LP}(c_1)) \geq z^*$} 
\Return $(x^*, y^*)$; 
\EndIIf  \label{endCritItems}

\State{\textit{Compute feasible BKP solutions:}}
\For {$i = 1,\dots, \gamma$}    \label{FeasGammaSol}
\If{$z(CRIT_2^{LP}(c_i)) < z^*$} 
$\hat{x} \leftarrow$ solve $CRIT_2(c_i)$;	
\If{$z(CRIT_2(c_i)) < z^*$} 
	$\hat{y} \leftarrow$ solve $KP(\hat{x})$;
	\If{$z(KP(\hat{x})) < z^*$}	$x^* = \hat{x}$, $y^* = \hat{y}$, $z^*= z(KP(\hat{x}))$;
	\EndIf 
\EndIf

\EndIf 
	
\EndFor \label{endFeasGammaSol}

\Comment{Step 2}	 
\State{\textit{Solve subproblems:}}	 
			
\For {all $c$ in list $L$}   \label{solveSub}
			
\IIf{$z(CRIT_2^{LP}(c)) \geq z^*$} 
\Return $(x^*, y^*)$;
\EndIIf 

			\State{Apply \eqref{eq:rxi}, \eqref{eq:rkj} and fix variables in $CRIT_2(c)$;}
	      	\State{$\bar{x} \leftarrow$ solve $CRIT_2(c)$;}	
					\While{$z(CRIT_2(c)) < z^*$}
					\State{$\bar{y} \leftarrow$ solve $KP(\bar{x})$;}
	\IIf{$z(KP(\bar{x}) < z^*$}	$x^* = \bar{x}$, $y^* = \bar{y}$, $z^*= z(KP(\bar{x}))$;
	\EndIIf 
	
	\State{Add constraints \eqref{changeone}, \eqref{interdone} to $CRIT_2(c)$;}  
  \State{$\bar{x} \leftarrow$ solve $CRIT_2(c)$;}	
					
					\EndWhile 
					\EndFor 		\label{endsolveSub}		 				
	\State{\Return  $(x^*, y^*)$.} \label{returnOpt}
\end{algorithmic}
\end{algorithm}

\section{Computational results}
\label{sec:ComRes}
All tests were performed on an Intel i7 CPU @ 2.4 GHz with 8 GB of RAM. The code was implemented in the C++ programming language.
The ILP solver used along the steps of the algorithm is CPLEX 12.6.2. 

The parameters of the ILP solver were set to their default values. The BKP instances with $n=35,40,45,50,55$ are generated in \cite{CaCarLoWo16} as follows. Profits $p_i$ and weights $w_i$ of the follower and weights $v_i$ of the leader are integers randomly distributed in $[1,100]$: 10 instances are 
generated for each value of $n$. The follower's capacity $C_l$ is set to  $\lceil (INS/11)\sum_i^{n}w_i\rceil$ where $INS$ $(=1, \dots, 10)$ denotes the instance identifier. 
The leader's capacity is randomly selected in the interval $[C_l - 10; C_l + 10]$.

We first tested our approach on these 50 benchmark instances. After some preliminary computational tests, we chose the following parameter entries for our approach: $\alpha = 100 $, $\beta = 100$, $\delta = 10$, $\mu = 150$, $\gamma = 2$. 
The corresponding results are presented in Table \ref{tab:CPUTime1}. For each instance, we report the optimal solution value, the CPU time to obtain an optimal solution and the number of subproblems explored in Step 2. The last column also reports the number of times model $CRIT_2(c)$ is solved along the two steps.

Algorithm CCLW in \cite{CaCarLoWo16} solves all instances with 50 items within a CPU time limit of 3600 seconds but runs out of time limit in instances 55-3, 55-4. Algorithm in \cite{FiLjMoSi18} solves all benchmark instances, requiring at most a computation time of about 85 seconds for solving instance 55-3. As the results in the table illustrate, the proposed exact approach outperforms the competing algorithms, successfully solving to optimality each instance in at most 1 second (this maximum CPU time is reached in instance 55-3)
with an average of 0.2 seconds. Also, the number of subproblems explored in Step 2 and the number of models $CRIT_2(c)$ solved are very limited. Notice that the tests in \cite{CaCarLoWo16} and in \cite{FiLjMoSi18} were carried out on different but comparable machines
in terms of hardware specifications.

\begin{table}[H]
	\centering
	\scriptsize
\begin{tabular}{|rr||*{1}{c|}*{1}{c|}*{1}{c|}*{1}{c|}}

  \hline
    &  & Optimal & CPU & \# Subprob. & \#  $CRIT_2(\cdot)$ \\ 
\textit{n} & \textit{INS} & Value &  Time & in Step 2 & solved  \\ \hline		

35 & 1 & 279 & 0.11 & 3 & 5  \\
   & 2 & 469  & 0.36 & 0 & 2  \\
   & 3 & 448 & 0.43 & 2 & 4  \\
   & 4 & 370 & 0.15 & 2 & 4  \\
   & 5 & 467 & 0.14 & 2 & 4  \\
   & 6 & 268 & 0.05 & 0 & 0  \\
   & 7 & 207 & 0.04 & 0 & 0  \\
   & 8 & 41 & 0.03 & 0 & 0  \\
   & 9 & 80 & 0.03 & 0 & 0  \\
   & 10 & 31 & 0.02 & 0 & 0  \\ \hline	
40 & 1 & 314 & 0.16 & 1 & 3  \\ 
   & 2 & 472 & 0.33 & 1 & 3  \\
   & 3 & 637 & 0.70 & 4 & 6  \\
   & 4 & 388 & 0.16 & 0 & 2  \\
   & 5 & 461 & 0.11 & 0 & 2  \\
   & 6 & 399 & 0.05 & 0 & 0  \\
   & 7 & 150 & 0.04 & 0 & 0  \\
   & 8 & 71  & 0.04 & 0 & 0  \\
   & 9 & 179 & 0.03 & 0 & 0  \\
   & 10 & 0 & 0.01 & 0 & 0  \\ \hline
45 & 1 & 427 & 0.21 & 3 & 5  \\
   & 2 & 633 & 0.36 & 1 & 3  \\
   & 3 & 548 & 0.61 & 3 & 5  \\
   & 4 & 611 & 0.27 & 1 & 3  \\
   & 5 & 629 & 0.22 & 2 & 4  \\
   & 6 & 398 & 0.06 & 0 & 0  \\
   & 7 & 225 & 0.04 & 0 & 0  \\
   & 8 & 157 & 0.04 & 0 & 0  \\
   & 9 & 53 & 0.03 & 0 & 0  \\
   & 10 & 110 & 0.02 & 0 & 0  \\ \hline
50 & 1 & 502 & 0.35 & 5 & 7  \\
   & 2 & 788 & 0.52 & 1 & 3  \\
   & 3 & 631 & 0.28 & 2 & 4  \\
   & 4 & 612 & 0.22 & 0 & 2  \\
   & 5 & 764 & 0.18 & 0 & 2  \\
   & 6 & 303 & 0.06 & 0 & 0  \\
   & 7 & 310 & 0.05 & 0 & 0  \\
   & 8 & 63 & 0.04 & 0 & 0  \\
   & 9 & 234 & 0.04 & 0 & 0  \\
   & 10 & 15 & 0.03 & 0 & 0  \\ \hline
55 & 1 & 480  & 0.37 & 3 & 5  \\
   & 2 & 702 & 0.31 & 1 & 3  \\
   & 3 & 778 & 1.11 & 8 & 10 \\
   & 4 & 889 & 0.56 & 5 & 7  \\
   & 5 & 726 & 0.09 & 0 & 0  \\
   & 6 & 462 & 0.07 & 0 & 0  \\
   & 7 & 370 & 0.06 & 0 & 0  \\
   & 8 & 387 & 0.05 & 0 & 0  \\
   & 9 & 104 & 0.04 & 0 & 0  \\ 
   & 10 & 178 & 0.03 & 0 & 0  \\ \hline
\end{tabular}
		\caption{BKP instances from \cite{CaCarLoWo16}.}
		\label{tab:CPUTime1}
\end{table}

\medskip
The computational tests in both \cite{CaCarLoWo16} and \cite{FiLjMoSi18} are limited to instances with 55 items. We then tested larger instances with $n=100, 200, 300, 400, 500$ according to the generation scheme in \cite{CaCarLoWo16}.
For each value of $n$ and $INS$, we generated 10 instances for a total of 500 instances. For these large instances, we set the parameters of our algorithm to the following values: $\alpha = 500 $, $\beta = 500$, $\delta = 20$, $\mu = 1000$, $\gamma = 5$. 
It is pointed out in \cite{CaCarLoWo16} that in instances with $INS \geq 5$ the follower's capacity constraint is expected to be inactive for any maximal 
leader's interdiction strategy. This makes these instances easy to solve. 
Our computational experiments confirm this trend also on larger instances: the proposed algorithm solves each instance with $n$ from 100 to 500 and $INS \geq 5$ in at most 8 seconds without never invoking Step 2. 
In the light of this consideration, we report in the following Table \ref{tab:CPUTime2} only the results for instances with $INS \leq 4$.  
\begin{table}[H]
	\centering
\begin{tabular}{|rr|*{1}{c|}*{2}{c|}*{2}{c|} *{2}{c|}}

  \hline
    &  & & \multicolumn{2}{|c|}{CPU} & \multicolumn{2}{|c|}{\# Subproblems} & \multicolumn{2}{|c|}{\#  $CRIT_2(\cdot)$} \\ 
    &  & & \multicolumn{2}{|c|}{Time} & \multicolumn{2}{|c|}{in Step 2} & \multicolumn{2}{|c|}{solved}  \\ \hline
								\textit{n}			&		\textit{INS}	& \#Opt  & Average & Max & Average & Max&Average&Max\\																			\hline		 
100 & 1 & 10 & 2.1  & 3.0  & 0.7  & 2.0  & 4.8  & 7.0  \\
& 2 & 10 & 5.6  & 9.9 & 3.8  & 9.0  & 8.9  & 16.0 \\
& 3 & 10 & 4.3  & 6.4  & 2.5  & 7.0  & 7.5  & 12.0 \\
& 4 & 10 & 2.3  & 4.5  & 0.7  & 4.0  & 5.2  & 9.0  \\ \hline	
200 & 1 & 10 & 5.3  & 10.7  & 3.4  & 7.0  & 8.9  & 17.0 \\
& 2 & 10 & 7.8  & 12.2 & 5.0  & 9.0  & 10.1 & 14.0 \\
& 3 & 10 & 9.1  & 13.6 & 6.4  & 12.0 & 12.3 & 19.0 \\
& 4 & 10 & 6.0  & 8.6  & 3.5  & 8.0  & 8.3  & 13.0 \\ \hline	
300 & 1 & 10 & 6.4  & 8.3  & 3.9  & 8.0  & 9.0  & 13.0 \\
& 2 & 10 & 15.5 & 37.4 & 7.2  & 14.0 & 13.5 & 23.0 \\
& 3 & 10 & 14.0 & 17.7 & 10.9 & 15.0 & 16.8 & 24.0 \\
& 4 & 10 & 8.7  & 13.2 & 4.9  & 11.0 & 9.9  & 16.0 \\ \hline	
400 & 1 & 10 & 8.8  & 12.3 & 6.7  & 10.0 & 12.8 & 17.0 \\
& 2 & 10 & 15.2 & 18.7 & 9.1  & 12.0 & 15.1 & 20.0 \\
& 3 & 10 & 19.0 & 30.5 & 12.0 & 17.0 & 18.8 & 32.0 \\
 & 4 & 10 & 12.6 & 16.5 & 8.4  & 23.0 & 13.8 & 30.0 \\ \hline	
500 & 1 & 10 & 11.9 & 18.2 & 7.6  & 13.0 & 13.1 & 20.0 \\
 & 2 & 10 & 20.6 & 26.6 & 11.0 & 20.0 & 17.0 & 25.0 \\
 & 3 & 10 & 21.2 & 25.8 & 12.7 & 17.0 & 17.8 & 22.0 \\
 & 4 & 10 & 15.1 & 17.1 & 4.7  & 8.0  & 9.8  & 13.0		\\	\hline	
\end{tabular}
		\caption{BKP instances with $n = 100, 200, 300, 400, 500$ and $INS \leq 4$.}
		\label{tab:CPUTime2}
\end{table}
The results in the table are summarized in terms of average, maximum CPU time and number of optimal solutions obtained with a time limit of 60 seconds. Similarly as in Table \ref{tab:CPUTime1}, we also report the average and maximum number of subproblems explored in Step 2, and the  average and maximum number of times model $CRIT_2(c)$ is solved.
The results 
illustrate the effectiveness of our approach. All instances are 
solved to optimality requiring 37.4 seconds at most for an instance with 300 items. The number of subproblems handled by Step 2 is in general limited, reaching a maximum value of 23 (in an instance with 400 items). Also,
the number of models $CRIT_2(c)$ to be solved is generally limited
and never superior to 32. We finally point out that the number of constraints \eqref{changeone}-\eqref{interdone} 
added to each subproblem is also limited: 
in the tested 
instances, 
the while--loop of Step 2 is executed 8 iterations at most.

\subsection*{Acknowledgments}
We thank M. Carvalho for providing us the benchmark instances of \cite{CaCarLoWo16}.

\section{Concluding remarks}
\label{sec:Concl}

We proposed for the Bilevel Knapsack with Interdiction Constraints a new exact approach which outperforms the state-of-the-art algorithms available in the literature. The algorithm relies on a new lower bound derived for the problem, which is improved by exploiting the expected features of an optimal solution of the classical knapsack problem. In future research, it will be worthy on one hand to investigate different correlations between profits and weights of the items in the follower's knapsack problem and on the other hand to which extent 
the proposed approach could be generalized to other bilevel optimization problems.


\end{document}